\documentclass[letterpaper, 10 pt, conference]{ieeeconf}  

\IEEEoverridecommandlockouts                              
\usepackage{amsmath,amsfonts}
\usepackage{algorithmic}
\usepackage{array}
\usepackage[caption=false,font=normalsize,labelfont=sf,textfont=sf]{subfig}
\usepackage{textcomp}
\usepackage{stfloats}
\usepackage{url}
\usepackage{verbatim}
\usepackage{graphicx}
\def\BibTeX{{\rm B\kern-.05em{\sc i\kern-.025em b}\kern-.08em
    T\kern-.1667em\lower.7ex\hbox{E}\kern-.125emX}}
\usepackage{balance}
\usepackage{amsmath}
\usepackage{amssymb}
\usepackage{cite}
\usepackage{stfloats}
\usepackage{amsthm}
\usepackage{indentfirst}
\usepackage{url}
\usepackage{array}
\usepackage[utf8]{inputenc}
\usepackage[english]{babel} 
\usepackage{graphicx}
\usepackage{epstopdf}
\graphicspath{ {images/} }
\usepackage{float}
\usepackage{tabularx}
\usepackage{booktabs}
\usepackage{wasysym}
\usepackage{multirow}
\usepackage{mathrsfs}
\usepackage{enumerate}
\usepackage{color}
\usepackage{slashbox}
\usepackage[ruled,vlined]{algorithm2e}
\usepackage{bm}
\usepackage{soul}
\usepackage{pdfpages}

\theoremstyle{plain}
\newtheorem{remark}{Remark}

\newtheorem{proposition}{Proposition}

\newtheorem{theorem}{Theorem}

\title{\LARGE \bf
Robust and Scalable Game-theoretic Security Investment Methods for Voltage Stability of Power Systems
}

\author{Lu An, Pratishtha Shukla, Aranya Chakrabortty, and Alexandra Duel-Hallen
\thanks{Dr. Lu An is with NVIDIA Corporation. Dr. Pratishtha Shukla is with Oak Ridge National Lab. Dr. Aranya Chakrabortty and Dr. Alexandra Duel-Hallen are with ECE Department, North Carolina State University. (e-mail: {\tt\small lan4@alumni.ncsu.edu}, {\tt\small pshukla@alumni.ncsu.edu}, {\tt\small achakra2@ncsu.edu}, {\tt\small sasha@ncsu.edu})}
}

\begin{document}

\maketitle
\thispagestyle{plain}
\pagestyle{plain}

\begin{abstract}
We develop investment approaches to secure electric power systems against load attacks where a malicious intruder (the attacker) covertly changes reactive power setpoints of loads to push the grid towards voltage instability while the system operator (the defender) employs reactive power compensation (RPC) to prevent instability. Extending our previously reported Stackelberg game formulation for this problem, we develop a robust-defense sequential algorithm and a novel genetic algorithm that provides scalability to large-scale power system models. The proposed methods are validated using IEEE prototype power system models with time-varying load uncertainties, demonstrating that reliable and robust defense is feasible unless the operator's RPC investment resources are severely limited relative to the attacker's resources.

\end{abstract}
\begin{keywords} Power Systems, Voltage stability, Load attacks, Game Theory, Security investment, Robust defense\end{keywords}

\section{Introduction}
Over the past decade, significant research has been performed on cybersecurity of electric power systems \cite{anu}, including security against various kinds of attacks on both generation and loads. With the proliferation of demand response and direct load control programs by utility companies across the United States, load attacks are becoming more common. Malicious attackers, for instance, can easily hack into the thermostats of domestic customers and covertly change their active and reactive power setpoints. When a large number of loads are manipulated in this way, the transmission grid can face a voltage collapse. In the literature so far, these attacks have mostly been studied from the point of view of detection and control \cite{rad}. Our objective in this paper is to formulate a RPC investment strategy that power system operators can adopt to secure the grid from this class of attacks, which may result in severe degradation of voltage stability \cite{simpson2016voltage}.

Game-theoretic methods, including Stackelberg games, have been widely applied to study security and optimization for wide ranges of cyber-physical systems, including power systems \cite{ Basar2019}. However, most game-theoretic investment approaches employ repeated games \cite{Basar2019}, which are not suitable when long-term, fixed security investment is desired. To address this issue, in \cite{an2020stackelberg} we developed a cost-based Stackelberg game (CBSG) to strategically allocate the players' long-term security investment resources. The cost-based Stackelberg equilibrium (CBSE) of this game not only optimizes the load attacker's and the system operator's payoffs, i.e., the increase and reduction of the voltage instability index \cite{simpson2016voltage}, respectively, but also saves their costs.

However, in \cite{an2020stackelberg} we assumed complete knowledge of the opponent's resources, which is idealistic for the defender, who acts first. In this paper, we develop and validate a robust-defense (i.e. robust-RPC) sequential method. Moreover, we assumed fixed values of constant-power loads in \cite{an2020stackelberg} while in this paper we consider a practical case of {\it time-varying} loads. Finally, the algorithm in \cite{an2020stackelberg} requires traversal search that has exponential complexity in the number of target loads. To address scalability of the CBSG, we develop an iterative, evolutionary, bidirectional, genetic algorithm (GA)-based method to find a CBSE, which improves upon previously investigated evolutionary methods for finding a Stackelberg equilibrium (SE) \cite{d2012equilibrium,vallee1999off} by utilizing parallel evolution for both players. We evaluate performance of the proposed methods on IEEE 9-bus and 39-bus power system models and show that reliable and robust defense is feasible unless the operator's RPC investment resources are severely limited relative to the load attacker's resources.

The rest of the paper is organized as follows. Section \ref{sec:model} summarizes the power system and the CBSG based on \cite{an2020stackelberg}. A bidirectional evolutionary algorithm is introduced in Section \ref{sec:BPEGA}. A robust-defense method is proposed in Section \ref{sec:RD}. Numerical results for IEEE prototype models are contained in Section \ref{sec:NR}, and Section \ref{sec:conclusion} concludes the paper.

\section{Power System Model and Cost-based Stackelberg Game}\label{sec:model}
We consider a power system with $G\geq 1$ generators and $K\geq 1$ constant power loads, where the load buses are indexed as the first $K$ buses, followed by $G$ generator buses. Let us denote the steady-state voltage magnitudes at the load buses as ${{\bm{V}}_L} = [V_1,\cdots,V_K] \in {\mathbb R^K}$ and at the generator buses as ${{\bm{V}}_G} = [V_{K+1},\cdots,V_{K+G}] \in {\mathbb R^G}$. The admittance matrix of the network is denoted as ${\bm{Y}}={\bm{G}}+j{\bm{B}}$. We partition the susceptance matrix ${\bm{B}} \in \mathbb R^{(K +G) \times (K +G)}$ into four block matrices as
\begin{equation}
\label{eq:B_c5}
{\bm{B}} = \left( {\begin{array}{*{20}{c}}
{{{\bm{B}}_{LL}}}&{{{\bm{B}}_{LG}}}\\
{{{\bm{B}}_{GL}}}&{{{\bm{B}}_{GG}}}
\end{array}} \right)
\end{equation}
where ${{\bm{B}}_{LL}}$ contains the interconnections among loads and ${{\bm{B}}_{LG}}={{\bm{B}}_{GL}}^{T}$ represents the interconnections between loads and generators. Following the derivations in \cite{simpson2016voltage}, one can then define the open-circuit load voltage vector as ${\bm{V}}_L^* =  - {\bm{B}}_{LL}^{ - 1}{{\bm{B}}_{LG}}{{\bm{V}}_G}$,
and, subsequently, the symmetric stiffness matrix as
\begin{equation}
\label{eq:Q_cirt_c5}
{{\bm{Q}}_{crit}} \triangleq \frac{1}{4}{\rm{diag}}({\bm{V}}_L^*) \cdot {{\bm{B}}_{LL}} \cdot {\rm{diag}}({\bm{V}}_L^*),
\end{equation}
where $\rm{diag}(\cdot)$ denotes the diagonal matrix.

\begin{table*}[!ht]
\caption{CBSG description based on \cite{an2020stackelberg}}
\vspace{-0.1in}
\centering
\begin{tabularx}{\textwidth}{@{\extracolsep{\fill}}|c|p{3.9in}|}
\hline
\textbf{Term} & \textbf{Definition}  \\
\hline
$\bm{a}=[a_k: k\in \{1,\cdots, K\}] \in \mathbb R^K$ & Actions of the load attacker\\
\hline
$a_k 
\in \{0,1/(L_a-1), \allowbreak 2/(L_a-1),\cdots,1\}$ & Attacker's investment level on load $k$ (chance of success)\\
\hline
$L_a$ & The number of attacker's investment levels\\
\hline
$q_a^k \leq q_a^{k,\max}, \forall k$ & Covertness constraint for the attacker on load $k$\\
\hline
$\bm{O}^j = [o_1^j, \cdots, o_k^j, \cdots, o_K^j], \forall j=1,\cdots,2^K$ & The $j^{\text{th}}$ outcome of attack at all loads\\
\hline
$P_{{\bm{a}}}(\bm{O}^j) = \prod\limits_{k:{\forall o_k^j} = 1}{{a_k}} \prod\limits_{k:{\forall o_k^j} = 0} {\left( {1 - {a_k}} \right)}$ & The probability of outcome $\bm{O}^j$\\
\hline
${\bm{q}}_a^j = \bm{O}^j \odot {{\bm{q}}_a}$ & The incremental reactive power demand for $\bm{O}^j$\\
\hline
$\gamma_a$ & The scaled cost of the attack on load $k$ at full effort level\\
\hline
$C_a=\gamma_a{||{{\bm{a}}}||_1}\leq 1$ & The total cost of the attacker\\
\hline
$\bm{d}=[d_k: k\in \{1,\cdots, K\}] \in \mathbb R^K$ & RPC actions of the defender \\
\hline
$d_k\in\{0,1/(L_d-1),2/(L_d-1),\cdots,1\}, \forall k\in \mathscr{L}_{ctrl}$ & Defender's investment (RPC) level on the control device of load $k$\\
\hline
$L_d$ & The number of defender's investment levels\\
\hline
$q_d^{k,\max}$ & The maximum reactive power the defender can compensate on load $k\in \mathscr{L}_{ctrl}$\\
\hline
$q_d^k= d_k q_d^{k,\max}$ & The defender's compensation on load $k\in \mathscr{L}_{ctrl}$\\
\hline
$\gamma_d$ & The scaled investment cost of RPC per protected load\\
\hline
$C_d=\gamma_d{||{{\bm{d}}}||_1}\leq 1$ & The defender's total investment cost\\
\hline
${{\bm{Q}}_L^{j}}={{{\bm{Q}}_L^{n}} + {\bm{q}}_a^j - {{\bm{q}}_d}}$ & Tthe reactive power demand vector for the $j^{\text{th}}$ outcome $\bm{O}^j$ given actions ${\bm{a}}$ and ${\bm{d}}$\\
\hline
$U_{j}({\bm{q}}_a^j, {{\bm{q}}_d}) = Clip\left( {{{\left\| {{\bm{Q}}_{crit}^{ - 1}{\bm{Q}}_L^{j}} \right\|}_\infty }}; (\Delta^n, 1) \right)$ & The voltage instability index (\ref{eq:delta_c5}) restricted to $[\Delta^n, 1]$ for the $j^{\text{th}}$ outcome $\bm{O}^j$ and actions ${\bm{a}}$, ${\bm{d}}$\\
\hline
$U^a({\bm{a}},{\bm{d}}) = \sum\limits_j^{{2^K}} {{P_{\bm{a}}}({{\bm{O}}^j})U_{j}({\bm{q}}_a^j, {{\bm{q}}_d})}$ & The attacker's expected utility given actions ${\bm{a}}$ and ${\bm{d}}$\\
\hline
$U^d({\bm{a}},{\bm{d}}) = - U^a({\bm{a}},{\bm{d}})$ & The defender's utility given actions ${\bm{a}}$ and ${\bm{d}}$\\
\hline
$(\bm{a}_{o}^*,\bm{d}_{o}^*)$ & Cost-based Stackelberg equilibrium (CBSE)\\
\hline
\end{tabularx}
\label{table:CBSG}
\vspace{-0.2in}
\end{table*}

Let ${{\bm{Q}}_L} = [Q_1,\cdots,Q_K] \in \mathbb R^K$ denote the $K$-dimensional real vector that represents the \textit{reactive power setpoints} at the load buses. Using (\ref{eq:Q_cirt_c5}), the \textit{voltage instability index} of the system is defined as
\begin{equation}
\label{eq:delta_c5}
\Delta  = ||{\bm{Q}}_{crit}^{ - 1}{{\bm{Q}}_L}|{|_\infty },
\end{equation}
which is easily computable and accounts for the structure of the entire grid topology. The $k^{\text{th}}$ entry of the matrix-vector product ${\bm{Q}}_{crit}^{ - 1}{{\bm{Q}}_L}$ captures the stability stress on load $k$, with $||\cdot||_\infty$ identifying the maximally stressed node. According to Theorem 1 in \cite{simpson2016voltage}, the power flow equation will have a unique, stable solution if $\Delta<1$. Equivalently, $\Delta\geq1$ indicates that at least one load bus in the model $i$ is overly stressed and can be responsible for a voltage collapse. We refer to $1-\Delta$ as the \textit{voltage stability margin} \cite{van2007voltage}. The larger the value of $\Delta$, the narrower the stability margin is and the closer the power system is to a voltage collapse.
Finally, let ${{\bm{Q}}_L^{n}}$ denote the system's \textit{nominal reactive power setpoint vector}. Then the \textit{nominal voltage instability index}  ${\Delta^{n}}$ is computed as the value of $\Delta$ (\ref{eq:delta_c5}) using ${{\bm{Q}}_L^{n}}$.

As proven in \cite[Supp. 6]{simpson2016voltage} and \cite{an2020stackelberg}, the voltage instability index $\Delta$ increases as the reactive power demands of the loads grow. The attacker can attempt to increase the reactive power demands at the load buses by breaking into the loads and adding an incremental vector
\begin{equation}
\label{eq:qa}
    {{\bm{q}}_a} = [q_a^1,\cdots,q_a^K] \in \mathbb R^K, 
\end{equation}
to ${{\bm{Q}}_L^{n}}$, thus driving ${\Delta}$ towards 1. Moreover, the attacker can make such load attacks covert by designing the entries of ${{\bm{q}}_a}$ small enough that they maintain the load bus voltages to be within their usual allowable range of 0.9 per unit (pu) to 1.1 pu while sill pushing ${\Delta}$ towards 1. As the attack is at the device level rather than at the system level, the state estimator placed at the local substation might be unable to detect it. To prepare for possible future attacks, the operator, or the defender, can switch on pre-installed voltage control devices, such as shunt capacitors and power electronic converters, to compensate for the potential increase in consumption in advance. We assume generators provide reactive power support as usual to maintain power balance, but extra support required to compensate for the attacks is provided by the control devices installed at a set of the load buses denoted $\mathscr{L}_{ctrl}=\{l_n: n=\{1,\cdots, N\}\}$, where $l_n$ is the load bus index and $N \leq K$. Thus, the $K$-dimensional RPC vector for all loads can be denoted as
\begin{equation}
\label{eq:qd}
    {{\bm{q}}_d}=[q_d^1,\cdots,q_d^K] \in \mathbb R^K, 
\end{equation}
where $q_d^k=0, \forall k\notin \mathscr{L}_{ctrl}$. If an attack is successful at each load, the overall reactive power balance becomes ${{\bm{Q}}_L}={{{\bm{Q}}_L^{n}} + {\bm{q}}_a - {{\bm{q}}_d}}$. The goal of the system operator (defender) is to strategically compensate for the attacker's actions and to avoid the voltage collapse by maintaining $\Delta$ as close as possible to the nominal $\Delta^{n}$.

The zero-sum CBSG in this paper is based on the CBSG in \cite{an2020stackelberg}. It is summarized in Table \ref{table:CBSG}. The attacker's and the operator's utilities are given by the expected value of $\Delta$ (\ref{eq:delta_c5}) restricted to $[\Delta^n, 1]$ and its opposite, respectively. The Cost-Based Backward Induction (CBBI) in \cite{an2020stackelberg} computes a CBSE, i.e., the load attack and RPC investment pair that provides the same players' payoffs as any SE but saves the attacker's and defender's costs. Theorem I in \cite{an2020stackelberg, robust-SG-shukla} and \cite[Appx.B]{An2020thesis} demonstrates existence of CBSE and other CBSG properties.

Finally, we assumed in \cite{an2020stackelberg} that the system model is given by the \textit{nominal} model.
However in practice, the reactive power setpoint vector ${{\bm{Q}}_L}$ is time-variant and is uncertain a priori. Consider a set of possible uncertain models $\{\mathscr{M}_i \;|\; i= 1,\cdots, M\}$ due to the fluctuation of real-time power consumption. When the nominal model is used by the operator and the attacker to compute their investment strategies $(\bm{a}_{o}^*,\bm{d}_{o}^*)$, there is a mismatch with the actual CBSE of the model $\mathscr{M}_i$. To evaluate the fractional difference of the utilities of the nominal model and uncertain model $\mathscr{M}_i$ at a CBSE of the nominal model, we 
\begin{equation}
\label{eq:mu_nom}
{\mu _{i}}\%  = \left| {\frac{{U^a({\bm{a}}_{o}^*,{\bm{d}}_{o}^*) - U_i^a({\bm{a}}_{o}^*,{\bm{d}}_{o}^*)}}{{U^a({\bm{a}}_{o}^*,{\bm{d}}_{o}^*)}}} \right| \times 100\%,
\end{equation}
where the payoffs of the system operator and the load attacker at $({\bm{a}}_o^*,{\bm{d}}_o^*)$ are given by $U_i^a({\bm{a}}_o^*,{\bm{d}}_o^*) = - U_i^d({\bm{a}}_o^*,{\bm{d}}_o^*)$, which are obtained from the ideal players' utilities in Table \ref{table:CBSG} by substituting the \textit{nominal} reactive power vector ${{\bm{Q}}_L^{n}}$ of the nominal model by the \textit{$i^{\text{th}}$ model's nominal reactive power vector} ${\bm{Q}}_L^{n,i}$ and replacing $\Delta^n$ by $\Delta^{n,i}$, the nominal voltage instability index of the \textit{$i^{\text{th}}$} model. Note that ${\mu_{nom}} = 0$.

\section{A Bidirectional Evolutionary Method for Computing a CBSE}\label{sec:BPEGA}
The CBBI algorithm in \cite{an2020stackelberg} is a traversal searching method with the complexity of ${\mathcal{O}}\left( {L_a^KL_d^N} \right)$. To reduce the computational complexity, we employ the following bidirectional parallel evolutionary GA-based (BPEGA) method (Algorithm 1). In Algorithm \ref{alg:BPEGA}, the population sizes of the each generation of the attacker's and defender's strategies are represented by even non-negative integers $S_a$ and $S_d$, respectively.
In Step 2 of Algorithm \ref{alg:BPEGA}, in each generation, each player determines the fitness value (performance metric) of the current population. 
For the defender's strategy candidate ${{\bm{d}}} \in POP^t_d$ where $POP^t_d$ denotes the defender's current strategy population, an attacker's best response within its strategy population $POP^t_a$ is given by ${g_{tmp}}({\bm{d}}) =  \mathop {\arg\max }\limits_{{\bm{a}} \in POP^t_a} U^a({\bm{a}},{\bm{d}})$.
Thus, the fitness value of each ${{\bm{d}}} \in POP^t_d$ is given by
\begin{equation}
\label{eq:fd}
fit_d (\bm{d}) =  U^d({g_{tmp}}({\bm{d}}),{\bm{d}}), \forall {{\bm{d}}} \in POP^t_d.
\end{equation}
Suppose that $\bm{d}'\in POP^t_d$ has the highest fitness value within the current population
\begin{equation}
\label{eq:d_prime}
\bm{d}' =  \mathop {\arg\max }\limits_{{\bm{d}} \in POP^t_d} fit_d (\bm{d}).
\end{equation}
For each value of $\bm{d}'$ that satisfies (\ref{eq:d_prime}), the attacker assigns the fitness value to all ${{\bm{a}}} \in POP^t_a$ as
\begin{equation}
\label{eq:fa}
fit_a (\bm{a}) = U^a(\bm{a},{\bm{d}'}), \forall {{\bm{a}}} \in POP^t_a.
\vspace{-0.2in}
\end{equation}
\setlength{\textfloatsep}{0pt}
\begin{algorithm}[!h]
\SetAlgoLined
\textbf{Parameter Initialization:} Population sizes $S_a$ and $S_d$, crossover probability $P_c$, mutation rate $P_m$, maximum number of generations $T$.
Current generation $t=0$\; 
\textbf{Step 1. Population Initialization:} Randomly selected feasible initial populations for both players $POP^0_a=\{{{\bm{a}}}_1, \cdots, {{\bm{a}}}_{S_a}\}$ and $POP^0_d=\{{{\bm{d}}}_1, \cdots, {{\bm{d}}}_{S_d}\}$, where $\forall \bm{a}\in POP^0_a$ and $\forall \bm{d}\in POP^0_d$  which satisfy the attacker's and defender's cost constraints, respectively\;
\While{the termination criteria are not satisfied,}{
\textbf{Step 2. Evaluation:} The defender and attacker compute ${U^d}({\bm{a}},{\bm{d}})$ and ${U^a}({\bm{a}},{\bm{d}})$ for all ${{\bm{d}}} \in POP^t_d$ and ${{\bm{a}}} \in POP^t_a$ and evaluate all individuals in the current generation to compute their fitness values based on (\ref{eq:fd}) and (\ref{eq:fa})\;
\For {Attacker and Defender}{
\textbf{Step 3. Selection:} Select $S_a/2$ (or $S_d/2$) pairs of parents $tmp_P^a$ (or $tmp_P^d$) using the Roulette Wheel selection method \cite{liu1998stackelberg}\;
\textbf{Step 4. Reproduction:} Apply crossover with probability $P_c$ and mutation operation with rate $P_m$ \cite{liu1998stackelberg} to generate $S_a$ (or $S_d$) children $tmp_c^a$ (or $tmp_c^d$)\;
\textbf{Step 5. Check feasibility:} For each individual in $tmp_c^a$ (or $tmp_c^d$), check if it is a feasible solution to attacker's (or defender's) cost constraint. Include all feasible children in the set $tmp_{c,f}^a$ (or $tmp_{c,f}^d$)\;
\textbf{Step 6. Combine and sort:} Combine the current generation $POP^t_a$ (or $POP^t_d$) with the set of feasible children $tmp_{c,f}^a$ (or $tmp_{c,f}^d$). Sort by the fitness value in the descending order. Sort the individuals with the same fitness values by their investment cost in ascending order. The $S_a$ (or $S_d$) individuals with the highest ranking are selected as the next generation $POP^{t+1}_a$ (or $POP^{t+1}_d$)\;
}
$t\leftarrow t+1$
}
\textbf{Step 7.} Apply the CBBI algorithm \cite{an2020stackelberg} to the final generation $POP^T_a$ and $POP^T_d$ to determine $({\bm{a}_o^*},{\bm{d}_o^*})$\;
\caption{Bidirectional Parallel Evolutionary Genetic Algorithm}\label{alg:BPEGA}
\end{algorithm}

In Step 6, if several individuals have the same fitness value and cost, the individuals who were selected in an earlier generation are placed ahead of those selected later. Since the ``combine and sort" process guarantees that the individual with the highest fitness value among the members of the current generation and of the feasible children set is selected for the next generation, the proposed BPEGA algorithm is an elitist GA \cite{deb2002fast}.

The iteration will stop when either of the following two conditions is satisfied: (1) For each player, all individuals in the current population are identical, i.e. ${\bm a}_i = {\bm a}_j, \forall {\bm a}_i, {\bm a}_j \in POP^t_a$ and ${\bm d}_i = {\bm d}_j, \forall {\bm d}_i, {\bm d}_j \in POP^t_d$; (2) The iteration has reached the preset maximum iteration number $T$.

Finally, a GA converges when a sequence of objective function evaluations approaches the maximum of the objective function as the number of iterations $T$ tends to infinity \cite{d2012equilibrium}. Since BPEGA employs parallel evolution of both players, the convergence result of \cite{vallee1999off}, which assumes only the defender's evolution, is not applicable to Algorithm \ref{alg:BPEGA}.

\vspace{-0.1in}
\begin{proposition}\label{prop-5-1}
Assume a crossover probability $P_c>0$ and a mutation rate $P_m>0$. As the number of iterations $T$ tends to infinity, both players' expected utilities generated by the BPEGA (Algorithm \ref{alg:BPEGA}) converge to the utilities at any SE and the strategy pair selected by the BPEGA has the lowest players' costs among all SEs of the CBSG.
\end{proposition}

\begin{proof}
The proof is based on \cite{robust-SG-shukla}, \cite[Prop.5.1]{An2020thesis} and is omitted due to the space constraints.
\end{proof}

The computational complexity of the BPEGA is ${\mathcal{O}}(T S_a S_d) \ll {\mathcal{O}}({L_a^K L_d^N})$, the complexity of the CBBI algorithm, when the system size is large and $T S_a S_d \ll {L_a^K L_d^N}$. In practice, power systems with different system sizes might require different $T$ values to achieve convergence as discussed in Sec.\ref{sec:NR}.

\section{Robust Defense against Load Attacks}\label{sec:RD}

In the CBBI algorithm of \cite{an2020stackelberg} or the BPEGA Algorithm \ref{alg:BPEGA}, the load attacker does not require the knowledge of the operator's cost per load $\gamma_d$ to determine its best response since it is the follower and thus observes the defender's RPC strategy before acting. However, the system operator acts first and thus relies on the knowledge of $\gamma_a$. When the defender does not have complete information about the attacker's resources, the proposed CBSG (Table {\ref{table:CBSG}}) is unsuitable. In the {\it robust-defense (RD)} algorithm described below, the defender employs a lower bound $\gamma_a^{est}$ on the attacker's actual cost $\gamma_a$, where $0 \leq \gamma_a^{est} \leq \gamma_a$, to compute its RPC strategy.

\noindent \textbf{Step 1:} \textbf{(a)} For each defender's RPC action ${\bm{d}}$ that satisfies the cost constraint, the defender estimates the set of attacker's best responses (i.e., load attack strategies) $\mathcal{G}(\gamma_a^{est}, {\bm{d}})$, where ${g}(\gamma_a^{est},{{\bm{d}}}) \in \mathcal{G}(\gamma_a^{est},{\bm{d}})$ if
\begin{align}
\label{eq:g_est_d}
{g}(\gamma_a^{est},{{\bm{d}}}) &= \mathop {\arg \max }\limits_{{\bm{a}}} {U^a}({{\bm{a}}},{{\bm{d}}}), \hfill  \\
 \mbox{s.t.}\;\;\;{\kern 1pt} & \gamma_a^{est}{||{{\bm{a}}}||_1} \le 1, \; q_a^k \leq q_a^{k,\max}, \forall k. \nonumber 
 \vspace{-0.1in}
\end{align}

\noindent\textbf{(b)} For each ${\bm{d}}$, the \textit{smallest-cost} estimated best response is
\begin{equation}
\label{eq:g_est_do}
g_o(\gamma_a^{est}, {\bm{d}}) = \mathop {\arg \min }\limits_{{g}(\gamma_a^{est},{\bm{d}}) \in \mathcal{G}(\gamma_a^{est},{\bm{d}})} ||{g}(\gamma_a^{est},{\bm{d}})||_1.
\end{equation}

\noindent\textbf{Step 2: (a)} The defender (operator) determines the set of its RPC strategies $\mathcal{D}_{est}$ where ${{\bm{d}}_{RD}^*} \in \mathcal{D}_{est}$ if
\begin{align}
\label{eq:d_RD_est}
{{\bm{d}}_{RD}^*} &= \mathop {\arg \max }\limits_{\bm{d}} {U^d}(g_o(\gamma_a^{est}, {\bm{d}}),{\bm{d}}),  \hfill\\
  \mbox{s.t.}\;\;\;{\kern 1pt} & {\gamma_d}||{{\bm{d}}}||_1 \le 1. \nonumber
  \vspace{-0.1in}
\end{align}

\noindent\textbf{(b)} A strategy in (\ref{eq:d_RD_est}) with the \textit{smallest cost} is chosen
\begin{equation}
\label{eq:do_RD_est}
{{\bm{d}}_{RD}^{o}} = \mathop {\arg \min }\limits_{{{{\bm{d}}_{RD}^*}}\in \mathcal{D}_{est}} ||{{\bm{d}}_{RD}^*}||_1.
\end{equation}

\noindent\textbf{Step 3: (a)} By observing the defender's RPC action ${{\bm{d}}_{RD}^{o}}$, the attacker finds its actual set of best responses $\mathcal{G}(\gamma_a, {{\bm{d}}_{RD}^{o}})$, where the load attack strategy ${{{\bm{a}}_{RD}^*}} \in \mathcal{G}(\gamma_a, {{\bm{d}}_{RD}^{o}})$ if
\begin{align}
\label{eq:g_do}
{{{\bm{a}}_{RD}^*}} &= {g}(\gamma_a,{{\bm{d}}_{RD}^{o}}) = \mathop {\arg \max }\limits_{{\bm{a}}} {U^a}({{\bm{a}}},{{\bm{d}}_{RD}^{o}}), \hfill  \\
 \mbox{s.t.}\;\;\;{\kern 1pt} & \gamma_a{||{{\bm{a}}}||_1} \le 1; q_a^k \leq q_a^{k,\max}, \forall k, \nonumber
 \vspace{-0.1in}
\end{align}

\noindent\textbf{(b)} If multiple solutions exist in $\mathcal{G}(\gamma_a, {{\bm{d}}_{RD}^{o}})$, the attacker chooses a load attack strategy with the \textit{smallest cost}
\begin{equation}
\label{eq:ao_RD}
{{\bm{a}}_{RD}^{o}} = {g_o}(\gamma_a,{{\bm{d}}_{RD}^{o}}) = \mathop {\arg \min }\limits_{{{{\bm{a}}_{RD}^*}}\in \mathcal{G}(\gamma_a, {{\bm{d}}_{RD}^{o}})} ||{{\bm{a}}_{RD}^*}||_1.
\end{equation}
A strategy pair $({{\bm{a}}_{RD}^{o}}, {{\bm{d}}_{RD}^{o}})$ is a \textit{cost-based RD solution}. Note that the actual defender's payoff differs from its estimate ${U^d}(g_o(\gamma_a^{est}, {{\bm{d}}_{RD}^{o}}),{{\bm{d}}_{RD}^{o}})$ and is given by ${U^d}({{\bm{a}}_{RD}^{o}},{{\bm{d}}_{RD}^{o}})$. 

{\theorem \label{thm:RD} \ \\
In the proposed RD method (steps 1$\sim$3 above): \\
\textbf{(a)} The actual utility of the defender is at least as large as its estimated utility, i.e.
\begin{equation}
\label{eq:thm2b_Ud}
{U^d}({{\bm{a}}_{RD}^{o}},{{\bm{d}}_{RD}^{o}}) \geq {U^d}(g_o(\gamma_a^{est}, {{\bm{d}}_{RD}^{o}}),{{\bm{d}}_{RD}^{o}}).
\end{equation}
\textbf{(b)} The defender's actual payoff ${U^d}({{\bm{a}}_{RD}^{o}},{{\bm{d}}_{RD}^{o}})$ increases with its estimate of the attacker's cost $\gamma_a^{est}$ and approaches its payoff $U^d({\bm{a}}_{o}^*,{\bm{d}}_{o}^*)$ at a CBSE of the CBSG in Table \ref{table:CBSG} as $\gamma_a^{est}$ tends to the actual $\gamma_a$ value.
}
\begin{proof} \
The proof is referred to \cite{An2023CDC-Sup} and is omitted due to the space constraints.
\end{proof}

{\remark
Theorem \ref{thm:RD}(a) demonstrates that assuming the worst-case attack scenario provides robust RPC when the system operator is uncertain about the load attacker's resources. Moreover, from Theorem \ref{thm:RD}(b), we conclude that as the operator's knowledge of the attacker's cost improves, the defender's actual payoff of the RD solution increases and approaches the payoff of the ideal game.
}

To evaluate the defender's utility loss due to its uncertainty about the attacker's budget, we compute the mismatch (loss) of actual defender's utility using the RD method relative to that at a CBSE of the ideal CBSG. For each set of $\gamma_a^{est}$, $\gamma_a$, and $\gamma_d$ values, this mismatch is computed as
\begin{equation}
\label{eq:mu_rdsg}
{\mu _{RD}}\%   = 
\left| {\frac{{{U^d}({{\bm{a}}_{RD}^{o}},{{\bm{d}}_{RD}^{o}}) - U^d({\bm{a}}_{o}^*,{\bm{d}}_{o}^*)}}{{U^d({\bm{a}}_{o}^*,{\bm{d}}_{o}^*)}}} \right| \times 100\%.
\end{equation}

Finally, the evolutionary Algorithm \ref{alg:BPEGA} can be easily modified to perform the RD method and shown to converge to an RD solution. Moreover, as for the CBSG, the RD method is developed assuming the nominal model, but can be applied to any uncertain model $i$ and evaluated using a metric similar to (\ref{eq:mu_nom}). The details are omitted for brevity.

\section{Numerical Results}
\label{sec:NR}
In Sec.\ref{sec:game-9}, we validate computationally efficient BPEGA method proposed in Sec.{\ref{sec:BPEGA}} by comparing its performance with the traversal algorithm for the IEEE 9-bus system. Then in Sec.\ref{sec:game-39}, we employ GA methods to analyze performance of the IEEE 39-bus system.

\subsection{IEEE 9-bus System} \label{sec:game-9}
The IEEE 9-bus system has 6 load buses, which are all potential targets for the attacker. We assume that the set of buses with control devices installed is $\mathscr{L}_{ctrl}=\{4,5,6,8\}$. In the simulation, $q_a^{k,\max}$ is determined by the covertness constraint, and we set $q_d^{k,\max}=2$ pu, $\forall k$. It was verified that these compensations do not violate the $[0.9, 1.1]$ pu voltage range for any bus. The $M$ load-uncertain models are created by adding independent zero-mean Gaussian random variables to the components of the nominal reactive power setpoint vector $\bm{Q}_L^{n}$, i.e., the $i^{\text{th}}$ model's nominal reactive power setpoint vector is given by $\bm{Q}_L^{n,i}=\bm{Q}_L^{n}(1+\bm{\epsilon}^i)$, where $\bm{\epsilon}^i\sim N(0,{\sigma}^2)$, $i = \{1,\cdots,M\}$, and the standard deviation $\sigma=0.1$ \cite{bo2009probabilistic}. Note that the number of models $M=20$, which satisfies the criteria for the sample size given the margin of error (MOE) of 0.05 and the confidence interval (CI) of $95\%$.
\begin{figure}[h]
\vspace{-0.1in}
  \centering
    \includegraphics[width=0.45\textwidth]{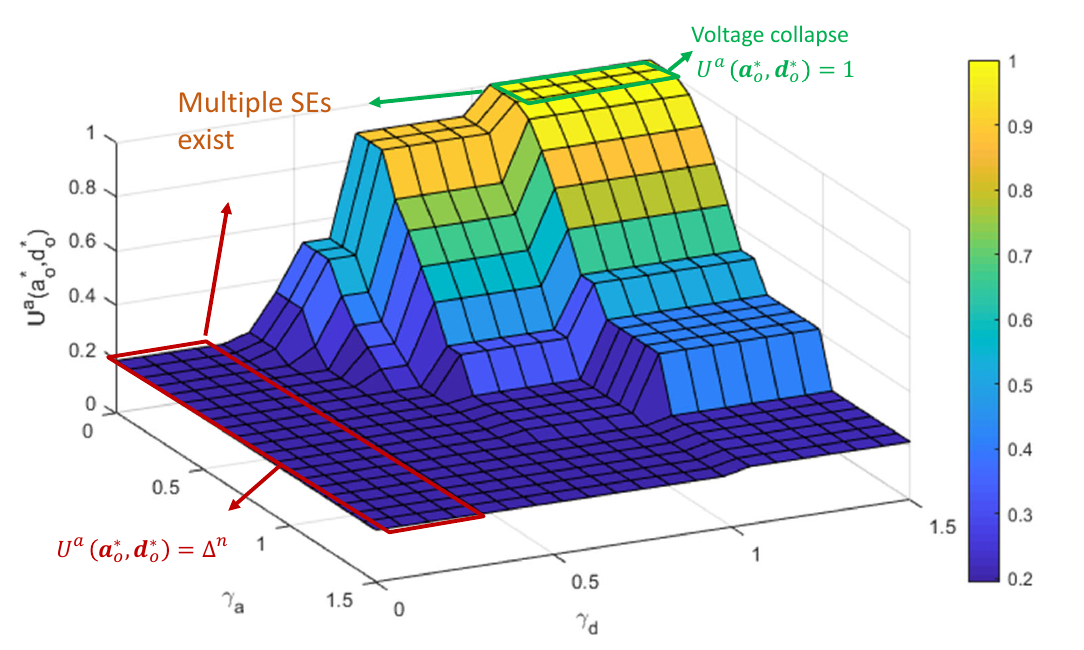}
    \vspace{-0.1in}
    \caption{Attacker's expected utility at CBSE vs. $\gamma_a$ and $\gamma_d$ for $L_a=L_d=3$, IEEE 9-bus system}
    \label{fig:Ua_avg_33}
\vspace{-0.1in}
\end{figure}
\begin{figure}[h]
\vspace{-0.1in}
  \centering
    \includegraphics[width=0.45\textwidth]{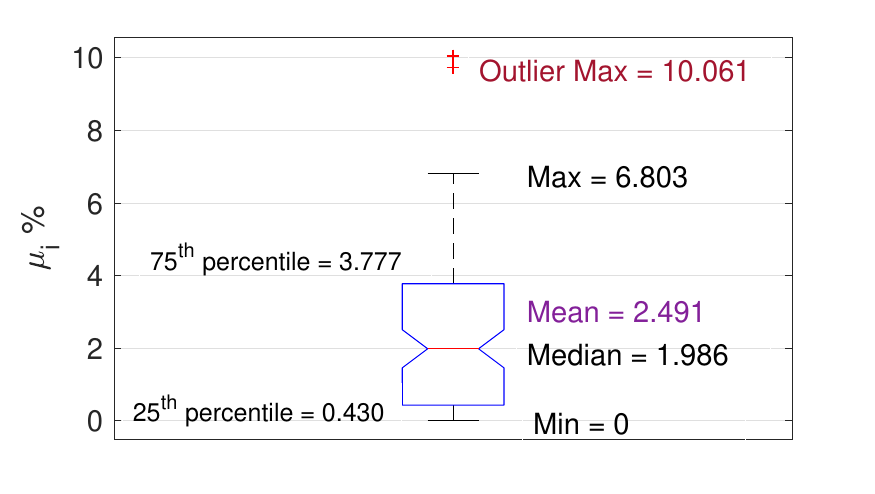}
    \vspace{-0.1in}
    \caption{Boxplot of the utility difference $\mu_{i}\%$ (\ref{eq:mu_nom}) over the randomly generated set of 20 models and 5 cost pairs ($\gamma_a$, $\gamma_d$): (0.1, 0.1), (0.1, 1.5), (0.75, 0.75), (1.5, 0.1), (1.5, 1.5); $L_a=L_d=3$; the IEEE 9-bus system}
    \label{fig:mismatch}
\end{figure}

Fig. \ref{fig:Ua_avg_33} shows the attacker's expected utility $U^a({\bm{a}_o^*},{\bm{d}_o^*})$ at a CBSE for varying $\gamma_a$ and $\gamma_d$ assuming a fixed nominal model. These trends are consistent with Theorem 1 in \cite{an2020stackelberg}. In particular, each player's utility improves as that player's cost decreases while the opponent's cost is fixed, and both players target ``important" loads with the greatest impact on $\Delta_0$, but the attacker tends to avoid the loads protected by the defender \cite{an2020stackelberg}. Fig. \ref{fig:mismatch} represents the utility difference (\ref{eq:mu_nom}) statistics for five different cost pairs. We observe that most uncertain models experience modest utility differences from the CBSE shown in Fig. \ref{fig:Ua_avg_33}, demonstrating robustness of CBSG to load uncertainty in the IEEE 9-bus power system.

Next, we validate convergence of BPEGA Algorithm \ref{alg:BPEGA}. We set $S_a=30$, $S_d=20$, $P_c=0.85$, $P_m=0.05$, and $T=30$. These initialization parameters are selected experimentally and reflect a trade-off between convergence and computational complexity for the IEEE 9-bus system.
For the nominal system model, Fig. \ref{fig:BPEGA_conv} shows that the BPEGA Algorithm \ref{alg:BPEGA} converges to a CBSE obtained by the traversal CBBI algorithm in fewer than 15 iterations, thus confirming Proposition 1. Similar results were obtained for other cost pairs, demonstrating fast convergence of the BPEGA Algorithm \ref{alg:BPEGA} \cite{An2020thesis}.
\begin{figure}[h]
\vspace{-0.1in}
  \centering
    \includegraphics[width=0.45\textwidth]{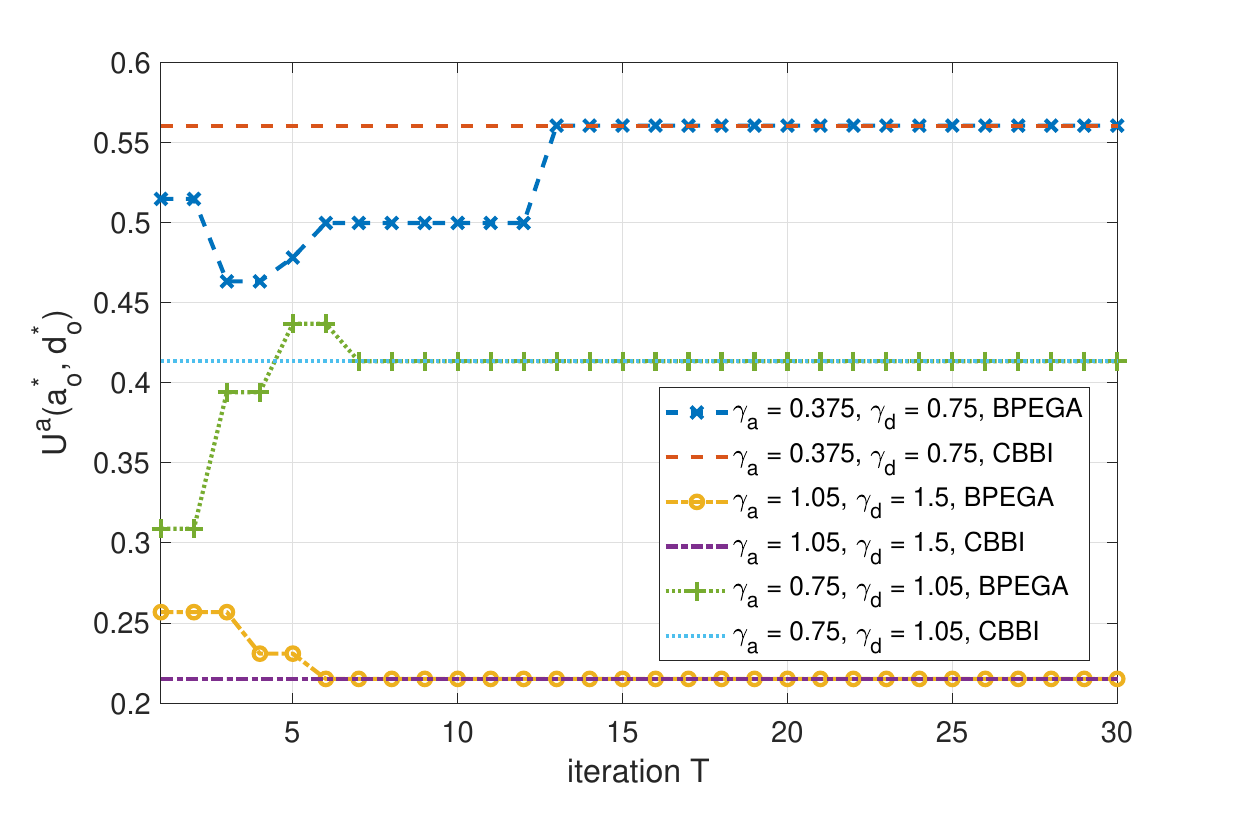}
    \vspace{-0.1in}
    \caption{Convergence of the attacker's utility in the BPEGA Algorithm \ref{alg:BPEGA} with $L_a=L_d=3$; the IEEE 9-bus system}
    \label{fig:BPEGA_conv}
    \vspace{-0.1in}
\end{figure}
\begin{figure}[h]
\vspace{-0.1in}
  \centering
    \includegraphics[width=0.45\textwidth]{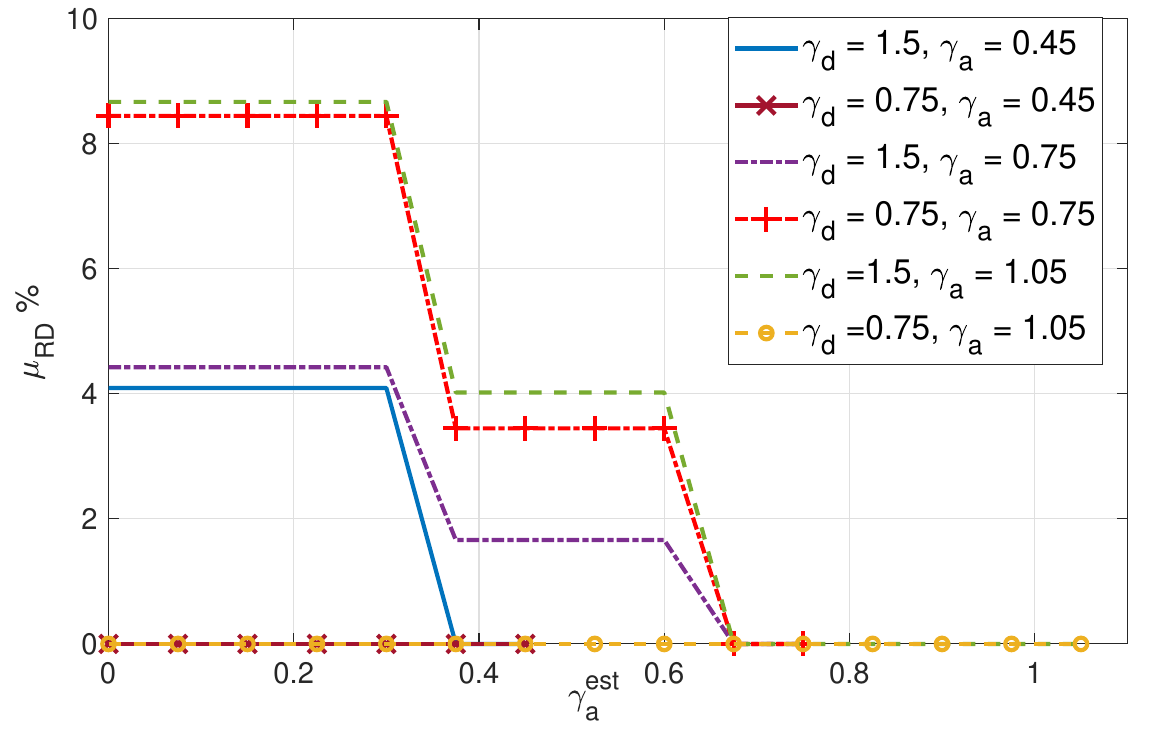}
    \vspace{-0.1in}
    \caption{The mismatch ${\mu_{RD}}\%$ (\ref{eq:mu_rdsg}) for varying $\gamma_a^{est}$ given $\gamma_a = \{0.45, 0.75, 1.05\}$, and $\gamma_d = \{0.75, 1.5\}$; $L_a=L_d=3$; the IEEE 9-bus system}
    \label{fig:RDSG_mismatch}
    \vspace{-0.1in}
\end{figure}

Finally, we compare the performance of the RD method for the nominal-model IEEE 9-bus system to that of the ideal CBSG by measuring the mismatch ${\mu _{RD}}\%$ (\ref{eq:mu_rdsg}) in Fig. \ref{fig:RDSG_mismatch} for varying costs $\gamma_a^{est}$ estimated by the system operator given selected $\gamma_a$ and $\gamma_d$ values. We observe that the mismatch reduces as the worst-case cost estimate $\gamma_a^{est}$ approaches the actual $\gamma_a$ value as shown in Theorem \ref{thm:RD}(b). We also found that the mismatch is zero for \textit{all} $\gamma_a^{est}$ values in many cases, e.g. when $(\gamma_a, \gamma_d) = (0.75, 0.45)$ or $(0.75, 1.05)$. Moreover, even when the most powerful attacker is assumed ($\gamma_a^{est} = 0$) while the actual $\gamma_a$ values are in the range $0 \leq \gamma_a \leq 1.5$ and $\gamma_d \in \{0.45, 0.75, 1.5\}$, the mismatch (\ref{eq:mu_rdsg}) for most cost pairs was less than 1\% with median of $\mu_{RD} = 0$ and 75\% percentile = 0.575, indicating that the RD method is robust to the system operator's uncertainty about the load attacker's resources. On the other hand, the defender's cost of the RD solution is at least as large as the defender's cost of a CBSE of the ideal CBSG. When the actual $\gamma_a$ is large, $\gamma_a^{est} = 0$, and $\gamma_d$ is relatively small, e.g. $\gamma_a \geq 1.05$ and $\gamma_d \leq 0.225$, the defender's overpayment using the RD method compared to the ideal CBSG can be as high as $300\%$ since the defender grossly overestimates the attacker's budget in the RD method. However, the defender's excess cost of the RD solution reduces as its cost per load $\gamma_d$ increases and $\gamma_a^{est}$ approaches the actual $\gamma_a$ value. Finally, the attacker's investment cost is the same in both methods.

 \vspace{-0.1in}
\subsection{IEEE 39-bus system}\label{sec:game-39}
Next, we apply the proposed methods to the IEEE 39-bus system, which contains 29 loads. We assume that the set of buses with control devices installed is $\mathscr{L}_{ctrl}=\{5,6,7,8,10,\allowbreak11,13\}$. This set includes both players' five most important loads \cite{an2020stackelberg}. The $M = 20$ load-uncertain models were created using the same method for the IEEE 9-bus system (Sec.\ref{sec:game-9}).
\begin{figure}[h]
\vspace{-0.2in}
  \centering
    \includegraphics[width=0.45\textwidth]{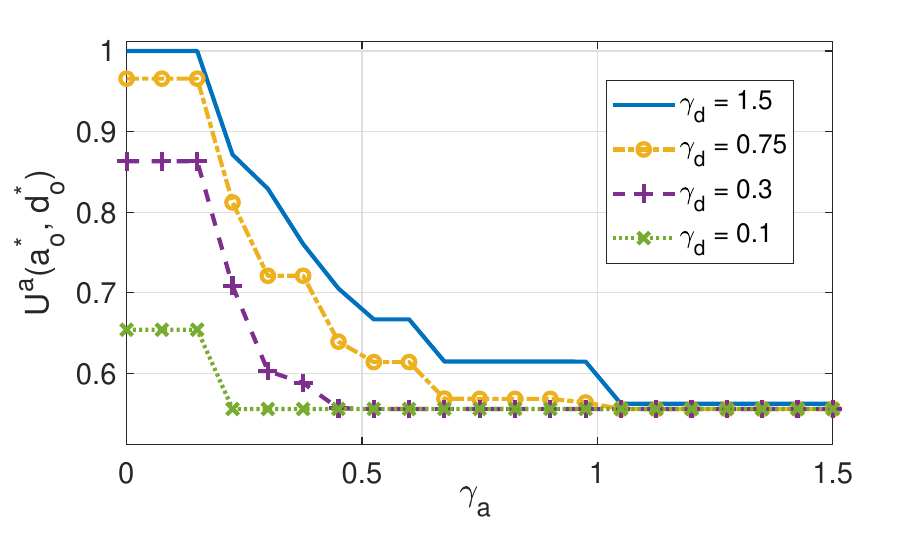}
    \vspace{-0.1in}
    \caption{Attacker's utility at a CBSE found by the BPEGA algorithm with $L_a=L_d=3$; the IEEE 39-bus system}
    \label{fig:U39_CBSE}
    \vspace{-0.1in}
\end{figure}

For levels of investment $L_a=L_d=3$, the complexities of the CBBI algorithm and the BPEGA Algorithm \ref{alg:BPEGA} are ${\mathcal{O}}\left( {3^{29}\times 3^{7}} \right)$ and ${\mathcal{O}}(T S_a S_d) = {\mathcal{O}}(30\times 30\times 20) = {\mathcal{O}}(18000)$, respectively. Since the former complexity is too high, we apply the BPEGA Algorithm \ref{alg:BPEGA} and the RD method for IEEE 39-bus system in Fig. \ref{fig:U39_CBSE} and \ref{fig:mismatch_rd_39}. We observe the same performance trends as in Fig. \ref{fig:Ua_avg_33}. Note that voltage collapse occurs only when the attacker has very small cost $\gamma_a$ while the defender's cost $\gamma_d \gg \gamma_a$. We conclude that voltage collapse can be successfully prevented in both IEEE 9-bus and 39-bus systems unless the defender's security resources are disproportionately limited relative to the attacker's budget.
\begin{figure}[H]
\vspace{-0.1in}
  \centering
    \includegraphics[width=0.45\textwidth]{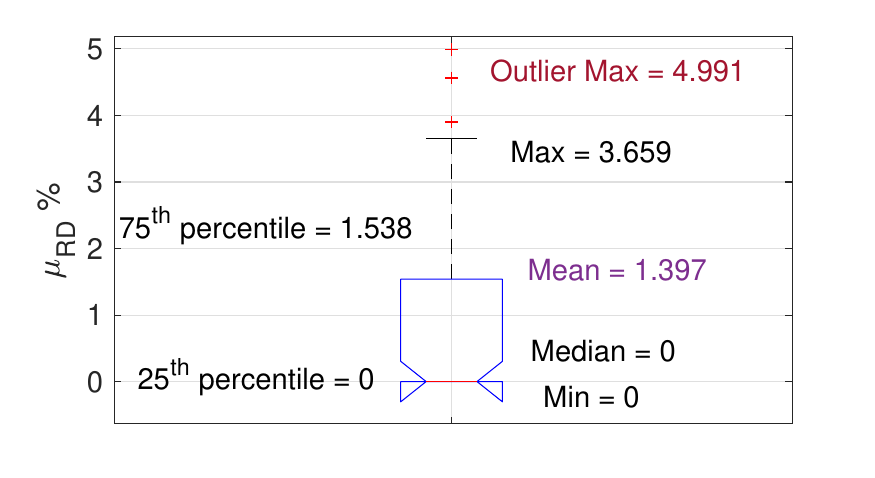}
    \vspace{-0.1in}
    \caption{Boxplot of the mismatch ${\mu_{RD}}\%$ (\ref{eq:mu_rdsg}) of the RD solution for $\gamma_a^{est} = 0$ over the cost pairs: $\gamma_a = [0:0.075:1.5]$, $\gamma_d = \{0.45, 0.75, 1.5\}$; $L_a = L_d = 3$; the IEEE 39-bus system}
    \label{fig:mismatch_rd_39}
    \vspace{-0.1in}
\end{figure}
Moreover, we found that for the IEEE 39-bus system, most uncertain models experience slight utility differences (\ref{eq:mu_nom}) ranging from $0\%$ to $1.836\%$, confirming the robustness of proposed CBSG to load uncertainty. Furthermore, Fig. \ref{fig:mismatch_rd_39} shows the boxplot for the mismatch (\ref{eq:mu_rdsg}) of the defender's utilities of the RD solution for selected defender's costs when the most powerful attacker is assumed ($\gamma_a^{est} = 0$) and the actual $\gamma_a$ values are in the range $0 \leq \gamma_a \leq 1.5$. We observe that the mismatch for most cost pairs is modest, confirming that the RD solution is robust to imperfect knowledge of the attacker's budget by the defender. Finally, for both IEEE 9-bus and 39-bus systems, we found that the RD method is robust to model uncertainty, with statistics similar to those in Fig. \ref{fig:mismatch}.

 \vspace{-0.1in}
\section{Conclusion}\label{sec:conclusion}
We investigated scalable, robust game-theoretic investment solutions for securing electric power systems from load attacks and associated voltage collapse. To address the scalability of the proposed methods to large power systems, a bidirectional, parallel, evolutionary generic algorithm (BPEGA) was developed. Moreover, we proposed a robust-defense (RD) method to address realistic scenarios where the defender lacks full information about the attacker's budget. It is demonstrated that the system operator is able to preserve voltage stability for both load- and/or information-uncertain scenarios unless its reactive power compensation resources are much more limited than the load attacker's resources.

\bibliographystyle{ieeetran}
\bibliography{ref_all, ref_SG, ref_power}

\begin{thebibliography}{10}
\providecommand{\url}[1]{#1}
\csname url@samestyle\endcsname
\providecommand{\newblock}{\relax}
\providecommand{\bibinfo}[2]{#2}
\providecommand{\BIBentrySTDinterwordspacing}{\spaceskip=0pt\relax}
\providecommand{\BIBentryALTinterwordstretchfactor}{4}
\providecommand{\BIBentryALTinterwordspacing}{\spaceskip=\fontdimen2\font plus
\BIBentryALTinterwordstretchfactor\fontdimen3\font minus
  \fontdimen4\font\relax}
\providecommand{\BIBforeignlanguage}[2]{{%
\expandafter\ifx\csname l@#1\endcsname\relax
\typeout{** WARNING: IEEEtran.bst: No hyphenation pattern has been}%
\typeout{** loaded for the language `#1'. Using the pattern for}%
\typeout{** the default language instead.}%
\else
\language=\csname l@#1\endcsname
\fi
#2}}
\providecommand{\BIBdecl}{\relax}
\BIBdecl

\bibitem{anu}
S.~M. Dibaji, M.~Pirani, D.~B. Flamholz, A.~M. Annaswamy, K.~H. Johansson, and
  A.~Chakrabortty, ``A systems and control perspective of cps security,''
  \emph{Annual Reviews in Control}, vol.~47, pp. 394 -- 411, 2019.

\bibitem{rad}
S.~{Amini}, F.~{Pasqualetti}, and H.~{Mohsenian-Rad}, ``Dynamic load altering
  attacks against power system stability: Attack models and protection
  schemes,'' \emph{IEEE Transactions on Smart Grid}, vol.~9, no.~4, pp.
  2862--2872, July 2018.

\bibitem{simpson2016voltage}
J.~W. Simpson-Porco, F.~D{\"o}rfler, and F.~Bullo, ``Voltage collapse in
  complex power grids,'' \emph{Nature communications}, vol.~7, p. 10790, 2016.

\bibitem{Basar2019}
S.~Etesami and T.~Basar, ``{Dynamic Games in Cyber-Physical Security: An
  Overview},'' \emph{Dynamic Games and Applications}, pp. 1--30, 2019.

\bibitem{an2020stackelberg}
L.~An, A.~Chakrabortty, and A.~Duel-Hallen, ``A stackelberg security investment
  game for voltage stability of power systems,'' in \emph{2020 IEEE 59th
  Conference on Decision and Control (CDC)}, 2020.

\bibitem{d2012equilibrium}
E.~D'Amato, E.~Daniele, L.~Mallozzi, and G.~Petrone, ``Equilibrium strategies
  via ga to stackelberg games under multiple follower's best reply,''
  \emph{International Journal of Intelligent Systems}, vol.~27, no.~2, pp.
  74--85, 2012.

\bibitem{vallee1999off}
T.~Vall{\'e}e and T.~Ba{\c{s}}ar, ``Off-line computation of stackelberg
  solutions with the genetic algorithm,'' \emph{Computational Economics},
  vol.~13, no.~3, pp. 201--209, 1999.

\bibitem{van2007voltage}
T.~Van~Cutsem and C.~Vournas, \emph{Voltage stability of electric power
  systems}.\hskip 1em plus 0.5em minus 0.4em\relax Springer Science \& Business
  Media, 2007.

\bibitem{robust-SG-shukla}
P.~Shukla, L.~An, A.~Chakrabortty, and A.~Duel-Hallen, ``A robust stackelberg
  game for cyber-security investment in networked control systems,'' \emph{IEEE
  Transactions on Control Systems Technology}, vol.~31, no.~2, pp. 856--871,
  2023.

\bibitem{An2020thesis}
\BIBentryALTinterwordspacing
L.~An, ``Game-theoretic methods for cost allocation and security in {Smart
  Grid},'' 2020, {Ph.D. thesis}. [Online]. Available:
  \url{https://repository.lib.ncsu.edu/handle/1840.20/38411}
\BIBentrySTDinterwordspacing

\bibitem{liu1998stackelberg}
B.~Liu, ``Stackelberg-nash equilibrium for multilevel programming with multiple
  followers using genetic algorithms,'' \emph{Computers \& Mathematics with
  Applications}, vol.~36, no.~7, pp. 79--89, 1998.

\bibitem{deb2002fast}
K.~Deb, A.~Pratap, S.~Agarwal, and T.~Meyarivan, ``A fast and elitist
  multiobjective genetic algorithm: Nsga-ii,'' \emph{IEEE transactions on
  evolutionary computation}, vol.~6, no.~2, pp. 182--197, 2002.

\bibitem{An2023CDC-Sup}
\BIBentryALTinterwordspacing
L.~An, P.~Shukla, A.~Chakrabortty, and A.~Duel-Hallen, ``Supplementary material
  for: Robust and scalable game-theoretic security investment methods for
  voltage stability of power systems,'' 2023, {supplementary document}.
  [Online]. Available: \url{https://tinyurl.com/supp-doc-cdc23-an}
\BIBentrySTDinterwordspacing

\bibitem{bo2009probabilistic}
R.~Bo and F.~Li, ``Probabilistic lmp forecasting considering load
  uncertainty,'' \emph{IEEE Transactions on Power Systems}, vol.~24, no.~3, pp.
  1279--1289, 2009.

\end{thebibliography}



\section*{Supplementary material}

# Supplementary Material for: "Robust and Scalable Game-theoretic Security Investment Methods for Voltage Stability of Power Systems"

Lu An, Pratishtha Shukla, Aranya Chakrabortty, and Alexandra Duel-Hallen

## I. PROOF OF PROPOSITION 1

**Proposition 1.** *Assume a crossover probability $P_c > 0$ and a mutation rate $P_m > 0$. As the number of iterations $T$ tends to infinity, both players' expected utilities generated by the BPEGA (Algorithm I) converge to the utilities at any SE and the strategy pair selected by the BPEGA has the lowest players' costs among all SEs of the CBSG.*

*Proof.* First, since the action spaces of both players are finite, from Theorem 1(a) in [1], an SE exists in the proposed game. Moreover, according to Theorem 1(b) in [1], all SEs provide the same pair of payoffs in a zero-sum game.

From Theorem 4.1 in [2], due to the crossover and mutation process, any feasible solution appears at least once in some generation with probability 1 when the set of feasible solutions is finite and the maximum number of iterations $T$ tends to infinity. Thus, since the action spaces of both players in the proposed game are finite, any feasible strategy pair $(\boldsymbol{a}, \boldsymbol{d})$ appears at least once in some generation with probability 1 when $T \to \infty$.

Denote the set of all SEs as $\mathscr{E} = \{(\boldsymbol{a}_i^*, \boldsymbol{d}_i^*) : i = 1, \cdots, S\}$. Suppose $(\boldsymbol{a}_o^*, \boldsymbol{d}_o^*)$ is an SE with the lowest players' costs (i.e. it satisfies (16) and (18) in [1]) and appear first in the evolution process among all lowest-cost SEs. Its fitness value $fit_d(\boldsymbol{d})$ for the defender must be the highest among the members of that generation since all SEs have the same utility for the defender. Moreover, since $\boldsymbol{d}_o^*$ has the smallest cost among all $\boldsymbol{d}_i^*$, it must have the highest ranking among the members of that generation due to the cost-sorting property in Step 6 of Algorithm 1. Thus, $\boldsymbol{d}_o^*$ will be included in the next generation. By a similar argument, $\boldsymbol{a}_o^*$ will also be retained in the following generation.

By the argument above, the SE pair $(\boldsymbol{a}_o^*, \boldsymbol{d}_o^*)$ will occur in all remaining generations of BPEGA. Thus, it will be included in the generations $POP_a^T$ and $POP_d^T$ as $T \to \infty$.

Finally, according to the Step 7 of the BPEGA algorithm, $(\boldsymbol{a}_o^*, \boldsymbol{d}_o^*)$ or another SE with the lowest players' costs will be produced by the BPEGA method. □

## II. PROOF OF THEOREM 1

**Theorem 1.**

*In the proposed RD method (steps 1∼3 in the Sec IV of the main paper):*
*(a) The actual utility of the defender is at least as large as its estimated utility, i.e.*

$$U^d(\boldsymbol{a}_{RD}^o, \boldsymbol{d}_{RD}^o) \geq U^d(g_o(\gamma_a^{est}, \boldsymbol{d}_{RD}^o), \boldsymbol{d}_{RD}^o). \tag{1}$$

*(b)* The defender's actual payoff $U^d(\boldsymbol{a}_{RD}^o, \boldsymbol{d}_{RD}^o)$ increases with its estimate of the attacker's cost $\gamma_a^{est}$ and approaches its payoff $U^d(\boldsymbol{a}_o^*, \boldsymbol{d}_o^*)$ at a CBSE of the CBSG in Table I as $\gamma_a^{est}$ tends to the actual $\gamma_a$ value.

*Proof.*

**(a)** Let $\Gamma_a$ and $\Gamma_a^{est}$ denote the action sets of the attacker when costs of attack per node are $\gamma_a$ and $\gamma_a^{est}$, respectively. Since $\gamma_a \geq \gamma_a^{est}$, $\Gamma_a \subseteq \Gamma_a^{est}$.

The attacker's best responses given to the defender's action $\boldsymbol{d}_{RD}^o$ and costs per node $\gamma_a$ and $\gamma_a^{est}$ are given by

$$\boldsymbol{a}_{RD}^o = g_o(\gamma_a, \boldsymbol{d}_{RD}^o) = \arg\max_{\boldsymbol{a} \in \Gamma_a} U^a(\boldsymbol{a}, \boldsymbol{d}_{RD}^o), \tag{2}$$

and

$$\boldsymbol{a}_{RD}^{est} = g_o(\gamma_a^{est}, \boldsymbol{d}_{RD}^o) = \arg\max_{\boldsymbol{a} \in \Gamma_a^{est}} U^a(\boldsymbol{a}, \boldsymbol{d}_{RD}^o), \tag{3}$$

respectively.

Since $\Gamma_a \subseteq \Gamma_a^{est}$, it is easy to show that

$$U^a(g_o(\gamma_a^{est}, \boldsymbol{d}_{RD}^o), \boldsymbol{d}_{RD}^o) \geq U^a(g_o(\gamma_a, \boldsymbol{d}_{RD}^o), \boldsymbol{d}_{RD}^o). \tag{4}$$

As the CBSG in [1] is zero-sum, i.e., $U^d(\boldsymbol{a}, \boldsymbol{d}) = -U^a(\boldsymbol{a}, \boldsymbol{d})$, it follows from (4) that

$$U^d(g_o(\gamma_a^{est}, \boldsymbol{d}_{RD}^o), \boldsymbol{d}_{RD}^o) \leq U^d(g_o(\gamma_a, \boldsymbol{d}_{RD}^o), \boldsymbol{d}_{RD}^o). \tag{5}$$

By substituting (2) into (5), we obtain (1).

**(b)** Suppose $\gamma_a^{'}$ and $\gamma_a^{est}$ are two defender's estimates of the attacker's cost where $\gamma_a^{est} \leq \gamma_a^{'} \leq \gamma_a$. Let $\Gamma_a^{est}$, $\Gamma_a^{'}$ and $\Gamma_a$ denote the action sets of the attacker when costs of attack per node are $\gamma_a^{est}$, $\gamma_a^{'}$ and $\gamma_a$, respectively. Since $\gamma_a^{est} \leq \gamma_a^{'} \leq \gamma_a$, $\Gamma_a \subseteq \Gamma_a^{'} \subseteq \Gamma_a^{est}$. Let $\Gamma_d$ denote the action set of the defender given its cost per node $\gamma_d$.

Denote the defender's actual payoffs when its estimates of the attacker's cost are $\gamma_a^{est}$ and $\gamma_a^{'}$ as $U^d(\boldsymbol{a}_{RD}^o, \boldsymbol{d}_{RD}^o)$ and $U^d(\boldsymbol{a}_{RD}^{'}, \boldsymbol{d}_{RD}^{'})$, respectively.

Since the game is zero-sum game, the defender's estimate of the attacker's best response when it assumes the attacker's cost $\gamma_a^{est}$ can be expressed as

$$g_o(\gamma_a^{est}, \boldsymbol{d}) = \arg\max_{\boldsymbol{a} \in \Gamma_a^{est}} U^a(\boldsymbol{a}, \boldsymbol{d}), \tag{6}$$

$$= \arg\max_{\boldsymbol{a} \in \Gamma_a^{est}} -U^d(\boldsymbol{a}, \boldsymbol{d}),$$

$$= \arg\min_{\boldsymbol{a} \in \Gamma_a^{est}} U^d(\boldsymbol{a}, \boldsymbol{d})$$

while assuming $\gamma_a^{'}$ results in the defender's estimate of the attacker's best response

$$g_o(\gamma_a^{'}, \boldsymbol{d}) = \arg\min_{\boldsymbol{a} \in \Gamma_a^{'}} U^d(\boldsymbol{a}, \boldsymbol{d}). \tag{7}$$

Combining eq.(12) and eq.(13) in the main paper with (6), the defender's action $\boldsymbol{d}_{RD}^o$ given the estimate $\gamma_a^{est}$ is

$$\boldsymbol{d}_{RD}^o = \arg\max_{\boldsymbol{d}\in\Gamma_d} U^d(g_o(\gamma_a^{est},\boldsymbol{d}),\boldsymbol{d}), \tag{8}$$

$$= \arg\max_{\boldsymbol{d}\in\Gamma_d}\min_{\boldsymbol{a}\in\Gamma_a^{est}} U^d(\boldsymbol{a},\boldsymbol{d}).$$

Similarly, the defender chooses its action $\boldsymbol{d}_{RD}^{'}$ given $\gamma_a^{'}$ where

$$\boldsymbol{d}_{RD}^{'} = \arg\max_{\boldsymbol{d}\in\Gamma_d} U^d(g_o(\gamma_a^{'},\boldsymbol{d}),\boldsymbol{d}), \tag{9}$$

$$= \arg\max_{\boldsymbol{d}\in\Gamma_d}\min_{\boldsymbol{a}\in\Gamma_a^{'}} U^d(\boldsymbol{a},\boldsymbol{d}).$$

After the attacker observes the defender's action $\boldsymbol{d}_{RD}^o$, it uses eq.(14) and eq.(15) in the main paper to find its best response

$$\boldsymbol{a}_{RD}^o = \arg\max_{\boldsymbol{a}\in\Gamma_a} U^a(\boldsymbol{a},\boldsymbol{d}_{RD}^o), \tag{10}$$

$$= \arg\min_{\boldsymbol{a}\in\Gamma_a} U^d(\boldsymbol{a},\boldsymbol{d}_{RD}^o),$$

$$= \arg\min_{\boldsymbol{a}\in\Gamma_a}\max_{\boldsymbol{d}\in\Gamma_d}\min_{\boldsymbol{a}\in\Gamma_a^{est}} U^d(\boldsymbol{a},\boldsymbol{d}).$$

Thus, the defender's actual payoff at the RD solution given its estimate on the attacker's cost $\gamma_a^{est}$ is

$$U^d(\boldsymbol{a}_{RD}^o,\boldsymbol{d}_{RD}^o) = \min_{\boldsymbol{a}\in\Gamma_a}\max_{\boldsymbol{d}\in\Gamma_d}\min_{\boldsymbol{a}\in\Gamma_a^{est}} U^d(\boldsymbol{a},\boldsymbol{d}). \tag{11}$$

Similarly, the defender's actual payoff given its estimate $\gamma_a^{'}$ is

$$U^d(\boldsymbol{a}_{RD}^{'},\boldsymbol{d}_{RD}^{'}) = \min_{\boldsymbol{a}\in\Gamma_a}\max_{\boldsymbol{d}\in\Gamma_d}\min_{\boldsymbol{a}\in\Gamma_a^{'}} U^d(\boldsymbol{a},\boldsymbol{d}). \tag{12}$$

From (11) and (12), since $\Gamma_a^{'}\subseteq\Gamma_a^{est}$, we conclude

$$U^d(\boldsymbol{a}_{RD}^{'},\boldsymbol{d}_{RD}^{'}) \geq U^d(\boldsymbol{a}_{RD}^o,\boldsymbol{d}_{RD}^o). \tag{13}$$

Finally, when the defender's estimate of the attacker's cost is the actual cost $\gamma_a$, the RD method is the same as the CBSG in Sec. III in the main paper. Thus the defender's actual payoff at the RD solution is $U^d(\boldsymbol{a}_o^*,\boldsymbol{d}_o^*)$ in this case.

Using $\Gamma_a \subseteq \Gamma_a^{'} \subseteq \Gamma_a^{est}$, we can expand (13) as

$$U^d(\boldsymbol{a}_o^*,\boldsymbol{d}_o^*) \geq U^d(\boldsymbol{a}_{RD}^{'},\boldsymbol{d}_{RD}^{'}) \geq U^d(\boldsymbol{a}_{RD}^o,\boldsymbol{d}_{RD}^o). \tag{14}$$

$\square$

## REFERENCES

[1] L. An, A. Chakrabortty, and A. Duel-Hallen, "A stackelberg security investment game for voltage stability of power systems," in *2020 IEEE 59th Conference on Decision and Control (CDC)*, *2020*.

[2] R. F. Hartl and R. Belew, "A global convergence proof for a class of genetic algorithms," *University of Technology, Vienna*, *1990*.


\end{document}